\documentclass[conference]{IEEEtran}
\IEEEoverridecommandlockouts
\usepackage{amsmath, amssymb, amsfonts}
\usepackage{graphicx}

\title{Minimization of Storage Cost in Distributed Storage Systems with Repair Consideration}

\author{
\IEEEauthorblockN{Quan Yu}
\IEEEauthorblockA{Department of Electronic Engineering\\
City University of Hong Kong\\
Email: quanyu2@student.cityu.edu.hk} \and
\IEEEauthorblockN{Kenneth W. Shum}
\IEEEauthorblockA{Institute of Network Coding\\
The Chinese University of Hong Kong\\
Email: wkshum@inc.cuhk.edu.hk} \and
\IEEEauthorblockN{Chi Wan Sung}
\IEEEauthorblockA{Department of Electronic Engineering\\
City University of Hong Kong\\
Email: albert.sung@cityu.edu.hk}
\thanks{This work was partially supported by a grant from the University Grants Committee of the Hong Kong Special Administrative Region, China (Project No. AoE/E-02/08).}
}

\newcommand{\SSS}{\mathsf{S}}
\newcommand{\In}{\mathsf{In}}
\newcommand{\Out}{\mathsf{Out}}
\newcommand{\DC}{\mathsf{DC}}
\newcommand{\flow}{\text{flow}}
\newcommand{\balpha}{\boldsymbol{\alpha}}

\newtheorem{theorem}{Theorem}

\begin{document}



%


\maketitle

\begin{abstract}
In a distributed storage system, the storage costs of different storage nodes, in general, can be different. How to store a file in a given set of storage nodes so as to minimize the total storage cost is investigated. By analyzing the min-cut constraints of the information flow graph, the feasible region of the storage capacities of the nodes can be determined. The storage cost minimization can then be reduced to a linear programming problem, which can be readily solved. Moreover, the tradeoff between storage cost and repair-bandwidth is established.
\end{abstract}

\IEEEpeerreviewmaketitle

\section{Introduction}
Distributed storage system provides an elegant way for reliable data storage. The storage nodes are distributed across a wide geographical area. When a small subset of storage nodes encounters a disaster, the source data object can still be reconstructed from the surviving nodes. To keep the reliability of the distributed storage system above a certain level, redundancy is essential. Two strategies are widely employed to introduce redundancy. The most straightforward strategy is replication, in which each storage node stores an entire copy of the source data object. This method, though simple, has low storage efficiency. The other strategy is erasure coding, adopted in Oceanstore~\cite{Oceanstore} and Total Recall~\cite{Totalrecall} systems. A source data object is divided into $k$ equal size fragments, and then these $k$ fragments would be encoded and distributed over $n$ storage nodes; each node stores one encoded fragment. As a result, the source data object can be reconstructed from any $k$ available storage nodes. Compared with the replication strategy, erasure coding provides better storage efficiency. However, in the face of repairing a failed storage node, erasure coding  wastes bandwidth. This is because a newcomer has to first reconstruct the entire source data object by downloading data from any $k$ surviving nodes and then to re-encode and store only a fraction of the downloaded data.

In order to minimize the repair-bandwidth, Dimakis et al. in~\cite{DGWR07,DGWR2010} propose the concept of {\em regenerating codes}. In their formulation, the data allocated to each storage node is equal to $\alpha$ units. When a node failure occurs, a newcomer chooses arbitrarily $d$ ($d\geq k$) available nodes to connect to and downloads $\beta$ units of data from each of these $d$ nodes. By introducing the information flow graph, they translate the repair problem into a single-source multi-cast problem in network coding theory. A tradeoff between the storage capacity per node and repair-bandwidth is also established. In~\cite{AKG10}, a distributed storage system, in which different download costs are associated with storage nodes, is introduced. Specifically, the authors focus on the scenario that there are totally two sets of storage nodes according to the different download costs. A tradeoff between download cost and repair-bandwidth is identified.

In most current studies of distributed storage systems, the amount of data stored on each node is simply assumed to be identical. How to  distribute the data across a collection of storage nodes is not an easy problem. Given the total storage budget, for different access models, Leong et al. in~\cite{LDH10} try to find the corresponding optimal storage allocation, in the sense of maximizing the probability of successful data recovery. It is shown that symmetric allocation is not always an optimal solution. However, its model deals with only the recovery problem of source data object; the repair problem of failed nodes is not considered.

In a realistic scenario, the storage nodes should be allowed to store different amounts of data according to the conditions of transmission links between source node and storage nodes as well as storage cost associated with each storage node. It is natural that different storage nodes may have different storage costs in a real distributed storage system. Since the storage nodes are distributed across a geographical wide area, the storage costs are affected by many factors, such as rents of the data storage centers, storage hardware costs and labor costs for maintenance.

In this paper, we combine the storage allocation problems with repair problems, and take different storage costs into consideration. Our objective is to seek an optimal storage allocation, which minimizes the total storage cost, subject to the constraints obtained by analyzing the corresponding information flow graphs. More specifically, we focus on the case that there are totally two types of storage nodes, each having a different storage cost. We will show that our storage cost minimization problem can be solved as a Linear Programming (LP) problem. By identifying the feasible region of this LP problem, the minimum storage cost would be obtained at the corner points. Moreover, the tradeoff between the storage cost and repair-bandwidth can also be established.

This paper is organized as follows. The problem of storage cost minimization is formulated in Section~\ref{sec:formulation}. In Section~\ref{sec:MinCut}, we draw the information flow graph, and identify the min-cut constraints. In Section~\ref{sec:LP}, we characterize the minimum storage cost by a linear programming problem. In Section~\ref{sec:curve}, we illustrate the tradeoff between storage cost and repair-bandwidth. We conclude in Section~\ref{sec:Conclusion}.

\section{Problem Formulation} \label{sec:formulation}
Consider a distributed storage system consisting of two types of storage nodes, each having a different storage cost per unit data. Let the storage cost for the first type of nodes be $C_1$, and the storage cost for the second type be~$C_2$. We assume that there are totally  $n$ storage nodes, among which $n_1$ nodes belong to type~1 and $n_2$ nodes belong to type~2.  A data object of size $M$ units is encoded and distributed among the $n$ storage nodes. For simplicity in presentation, we assume that the storage capacities of the nodes of type~1 are identical and equal to $\alpha_1$, while the storage capacities of type~2 nodes are identical and equal to $\alpha_2$.
The total storage cost for storing the original data object can be calculated as
$C_1n_1\alpha_1+C_2n_2\alpha_2$.

There are two components in the design of distributed storage systems:
(i) A data collector (DC) connecting to any $k$ available storage nodes should be able to reconstruct the original data object by downloading a number of packets from these $k$ storage nodes. (ii) Once a storage node fails, a newcomer initializes a repair process and regenerates the failed node so that any DC, connecting to this newcomer and other $k-1$ existing nodes, is able to rebuild the original data object. During the repair process, the newcomer chooses $d$ ($d\geq k$) surviving storage nodes to connect to, each belongs either to type~1 or type~2, and then downloads $\beta$ units of data from each of these $d$ nodes. The traffic $d\beta$ incurred by the repair operation is defined as the {\em repair-bandwidth}.

There are two modes for storage-node repair. The first one is called functional repair and the second one is exact repair. In {\em functional repair}, the content of the newcomer is not necessarily the same as the content in the failed node to be replaced. We only need to ensure that any DC connecting to any $k$ storage nodes is able to rebuild the original data file. In {\em exact repair}, the content of the newcomer is required to be exactly the same as the content in the failed node. We refer the readers to~\cite{RSKR09,SR11} for code construction for exact repair.
In this paper, we focus on functional repair.

We model the distributed storage system as an information flow graph introduced in~\cite{DGWR07,DGWR2010}. For any information flow graph, to be detailed in the next section, if the minimum of the cut capacities between the source and each data collector is not less than the object data size~$M$, then there always exists a linear network code such that all data collectors can reconstruct the data object~\cite{LYC03}.

Our objective of this work is to seek an optimal storage allocation across the $n$ storage nodes that minimizes the total storage cost $C_S$ under the constraints described above.

\section{Min-Cut Constraints} \label{sec:MinCut}

The distributed storage network with storage cost is abstracted and modeled by an information flow graph $G=(\mathcal{V},\mathcal{E})$. We label the storage nodes from 1 to $n$, so that the storage nodes 1 to $n_1$ are of type 1, while the storage nodes $n_1+1$ to $n$ are of type~2.

The vertices are divided into stages, starting from stage~$-1$.
In the $i$-th stage, we have one newcomer which replaces a failed node. The edges are directed, and labeled by the corresponding capacities. We define the information flow graph more formally as follows.

\begin{enumerate}
\item There is a single source vertex, $\SSS$, in stage~$-1$. It represents the data object to be distributed among the storage nodes.

\item We put $2n$ vertices in stage~0. These vertices are called $\In_i$ and $\Out_i$, for $i=1,2,\ldots, n$. For each $i$, we draw a directed edge from the source vertex to $\In_i$ with infinite capacity. For $i=1,2,\ldots, n_1$, we draw a directed edge from $\In_i$ to $\Out_i$ with capacity $\alpha_1$. This signifies that the storage capacities in the storage nodes of type 1 are limited to $\alpha_1$ units. For $i=n_1+1,n_1+2,\ldots, n$, we draw a directed edge from $\In_i$ to $\Out_i$ with capacity $\alpha_2$. This indicates that each node of type 2 can store no more than $\alpha_2$ units of data.

\item For $s=1,2,\ldots,$ we put two vertices in stage $s$. If storage node $i$ fails in the $s$-th stage, we construct two vertices, $\In_i$ and $\Out_i$ in stage~$s$.  The vertex $\In_i$ is connected to $d$ ``Out'' nodes in earlier stages. The capacities of these $d$ edges are all equal to $\beta$. If node $i$ is of type $j$, ($j$ is either 1 or 2) we draw an edge from $\In_i$ to $\Out_i$ with capacity $\alpha_j$.

\item A data collector is represented by a vertex, called $\DC$, which is connected to $k$ ``Out'' nodes with distinct subscripts. All these $k$ edges have infinite capacity.
\end{enumerate}

An example of the information flow graph is shown in Fig.~\ref{fi:flowgraph}.

\begin{figure}
  \centering
  \scalebox{0.5}{\includegraphics{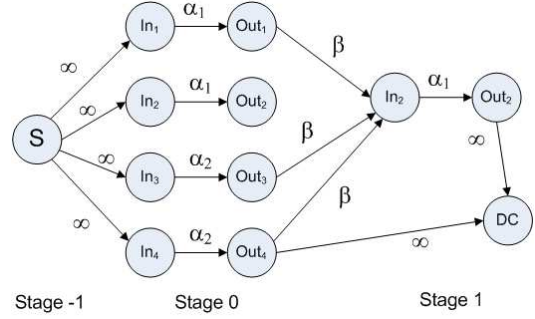}}
  \caption{Information Flow Graph ($n_1=n_2=2$, $d=3$, $k=2$).}
  \label{fi:flowgraph}
\end{figure}

A {\em flow} on the information flow graph $G$ is an assignment of non-negative real numbers to the edges, satisfying the flow conservation constraints and the capacity constraints. A flow $F$ can be regarded as a function from the edge set $\mathcal{E}$ to the set of non-negative real numbers, $F:\mathcal{E} \rightarrow \mathbb{R}_+$, such that

(i) for each edge $e\in\mathcal{E}$, $F(e)$ is less than or equal to the capacity of~$e$, and

(ii) for each vertex other than the source vertex and the data collectors, the sum of incoming flows is equal to the sum of outgoing flows, i.e., if $v\in\mathcal{V}$ is either an ``in'' or ``out'' vertex, then
\[
 \sum_{e: Head(e) =v} F(e) = \sum_{e: Tail(e) =v} F(e)
\]
where $Head(e)$ and $Tail(e)$ stand for the head and tail of edge $e$ respectively.

The {\em value} of a flow $F$ with respect to a data collector $\DC$ is defined as the sum of incoming flows to this data collector,
\[
 \sum_{e: Head(e) = \DC} F(e).
\]
The maximal flow value with respect to a specific data collector $\DC$, denoted by max-$\flow(\DC)$, is the maximal value of flow to this data collector $\DC$, over all legitimate flows. The max-flow theorem in network coding~\cite{LYC03, KM03} says that if max-$\flow(\DC) \geq M$ for all data collector $\DC$, then there exists a linear network code which sends $M$ units of data to every data collector.

Given a particular data collector $\DC$, an {\em $(\SSS,\DC)$-cut} is a partition of the vertices $(\mathcal{W}, \bar{\mathcal{W}})$ such that $\SSS\in\mathcal{W}$ and $\DC\in\bar{\mathcal{W}}$. (Here $\bar{\mathcal{W}}$ stands for the set complement of $\mathcal{W}$ in $\mathcal{V}$.) The {\em capacity} of an $(\SSS,\DC)$-cut is defined as the sum of capacities of the edges from $\mathcal{W}$ to $\bar{\mathcal{W}}$.
It is well known that the max-flow with respect to a data collector $\DC$ is equal to the minimum cut capacity. Let the capacity of an edge $e$ be denoted by $c(e)$. For each $(\SSS,\DC)$-cut , we have the following constraint
\begin{equation}
\sum_{\begin{subarray}{c}
        e\\  Tail(e)\in\mathcal{W} \\
        Head(e)\in\bar{\mathcal{W}}
      \end{subarray}} c(e) \geq M.
\label{eq:cut}
\end{equation}
The summation in \eqref{eq:cut} is over all edges with heads in $\mathcal{W}$ and tails in $\bar{\mathcal{W}}$. The storage cost minimization problem can be expressed as follows:
\begin{equation}
\min C_S\triangleq C_1n_1\alpha_1+C_2n_2\alpha_2,
\label{eq:objective}
\end{equation}
subject to the constraints \eqref{eq:cut} for all $(\SSS,\DC)$-cuts $(\mathcal{W},\bar{\mathcal{W}})$. The optimization is a linear programming problem with two variables $\alpha_1$ and $\alpha_2$.

Given parameters $n_1$, $n_2$, $k$, $d$, $M$, $\beta$, $C_1$ and $C_2$, we let the minimum storage cost in the above linear program be $C_S^*$. The values of $\alpha_1$ and $\alpha_2$ which achieve $C_S^*$ are denoted by $\alpha_1^*$ and $\alpha_2^*$. We will also investigate the tradeoff between the storage cost and the repair-bandwidth. In this context, we will write $C_S^*(\beta)$, $\alpha_1^*(\beta)$ and $\alpha_2^*(\beta)$ as functions of~$\beta$.

\begin{theorem}
Let $\mathcal{A}$ be the set of $k$-vectors
$$\boldsymbol{\alpha} = (\alpha(1), \alpha(2), \ldots, \alpha(k))$$
whose components are either $\alpha_1$ or $\alpha_2$, and the number of components in $\balpha$ which equal $\alpha_i$ is at most $n_i$, for $i=1,2$.
Given $n_1$, $n_2$, $k$, $d$ and $\beta$, the file size $M$ is upper bounded by
\begin{equation}
M \leq  \sum_{i=1}^{k} \min\{\alpha(i),(d-i+1)\beta \},
\label{eq:mincut}
\end{equation}
for any $\balpha \in \mathcal{A}$. Furthermore, we can construct an information flow graph such that equality in~\eqref{eq:mincut} holds for some $\balpha \in \mathcal{A}$.
\label{thm:mincut}
\end{theorem}

\begin{proof} (sketch)
The proof is based on the analysis of min-cut in the information flow graph, and is similar to the proof of~\cite[Lemma 2]{DGWR2010}. The main difference is that in this paper, the capacity of an edge between an ``in'' node and an ``out'' node may be either $\alpha_1$ or $\alpha_2$, whereas in~\cite{DGWR2010}, all $\alpha$'s are identical. Because the number of storage nodes of type~$i$ is equal to $n_i$ ($i=1,2$), there are at most $n_i$ edges with capacity $\alpha_i$ in a min-cut. Therefore we take the minimum only over the set $\mathcal{A}$. As the proof of~\eqref{eq:mincut} is basically the same as that of Lemma~2 in \cite{DGWR2010}, the details are omitted.
\end{proof}

\begin{figure}
  \centering
  \scalebox{0.4}{\includegraphics{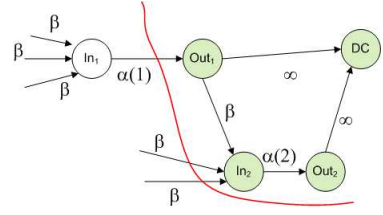}}
  \caption{An example of cut (d=3, k=2).}
  \label{fi:cut}
\end{figure}

We illustrate Theorem~\ref{thm:mincut} by the example in Fig.~\ref{fi:flowgraph}. A sample cut $(\mathcal{W},\bar{\mathcal{W}})$ is shown in Fig.~\ref{fi:cut}. The vertices in $\bar{\mathcal{W}}$ are drawn in shaded color. The values of $\alpha(1)$ and $\alpha(2)$ are either $\alpha_1$ or $\alpha_2$. The set $\mathcal{A}$ consists of four pairs
$(\alpha_1,\alpha_1)$, $(\alpha_1,\alpha_2)$, $(\alpha_2,\alpha_1)$, and $(\alpha_2,\alpha_2)$. The file size $M$ is upper bounded by
\begin{align*}
M & \leq  \min\{\alpha_1, 3\beta\} + \min\{\alpha_1,2\beta\} \\
M & \leq  \min\{\alpha_2, 3\beta\} + \min\{\alpha_1,2\beta\}\\
M & \leq  \min\{\alpha_1, 3\beta\} + \min\{\alpha_2,2\beta\}\\
M & \leq  \min\{\alpha_2, 3\beta\} + \min\{\alpha_2,2\beta\} .
\end{align*}

The cost minimization problem is to minimize $C_S$ in~\eqref{eq:objective}, subject to the constraints in~\eqref{eq:mincut} over all $\balpha \in \mathcal{A}$. This optimization can be reduced to a linear programming problem, as shown in the next theorem.

\begin{theorem} Let $\theta_m \triangleq (k-m)(2d-k-m+1)\beta/2.$
The cost minimization problem is equivalent to minimizing $C_S$ as defined in~\eqref{eq:objective} subject to the following $2(k+1)$ linear constraints,
\begin{align}
M &\leq  \min\{m,n_1\}\alpha_1 + (m-\min\{m,n_1\})\alpha_2 + \theta_m, \label{eq:LP1} \\
M &\leq  (m-\min\{m,n_2\})\alpha_1 + \min\{m,n_2\}\alpha_2 + \theta_m, \label{eq:LP2}
\end{align}
for $m=0,1,\ldots, k$.
\label{thm:LP}
\end{theorem}

\begin{proof}
For each $\balpha \in \mathcal{A}$, the inequality in~\eqref{eq:mincut} can be replaced by $2^k$ linear inequalities.
We introduce a ``switch'' function
\[
 s_b(x,y) \triangleq \begin{cases}
 x & \text{ if } b=0, \\
 y & \text{ if } b=1.
 \end{cases}
\]
Let $\mathcal{B}=\{0,1\}^k$ be the set of all binary vectors of length $k$. The inequality in~\eqref{eq:mincut} is equivalent to the following $2^k$ inequalities:
\[
 M \leq \sum_{i=1}^k s_{b_i}(\alpha(i), (d-i+1)\beta),
\]
where $(b_1,b_2,\ldots, b_k) \in\mathcal{B}$. This yields $|\mathcal{A}|2^k$ linear inequalities.

We may group these $|\mathcal{A}|2^k$ linear inequalities by the number of zeros in $(b_1,b_2,\ldots, b_k)$. Among those linear inequalities with $m$ zeros in $(b_1,b_2,\ldots, b_k)$, where $m$ is an integer between 0 and $k$, the most stringent inequality is the one associated with
$$(b_1,b_2,\ldots, b_k) = (\underbrace{0,0,\ldots, 0}_{m},\underbrace{1,1,\ldots, 1}_{k-m}),$$
which is,
\begin{align*}
 M &\leq \sum_{i=1}^{m} \alpha(i) + \sum_{i=m+1}^k (d-i+1)\beta \\
 &=\sum_{i=1}^{m} \alpha(i) + \theta_m.
\end{align*}
If there are $p$ $\alpha_1$'s and $q$ $\alpha_2$'s among $\alpha(1),\ldots, \alpha(m)$, we have
\[
M \leq p\alpha_1 + q \alpha_2 + \theta_m.
\]
Among the group of linear inequalities with $m$ zeros in $(b_1,b_2,\ldots, b_k)$, many inequalities are redundant, meaning that we can remove them without altering the feasible region. We only retain two inequalities, the one in which the coefficient of $\alpha_1$ is smallest, and the one in which the coefficient of $\alpha_2$ is smallest, namely the inequalities in~\eqref{eq:LP1} and~\eqref{eq:LP2}. The other inequalities in the same group are some convex combinations of these two inequalities, and hence can be ignored without changing the shape of the feasible region.
\end{proof}

If we put $m=0$ in either \eqref{eq:LP1} or \eqref{eq:LP2}, we see that there is no feasible solution to the linear programming problem if $\beta$ is strictly less than $\frac{2M}{k(2d-k+1)}$. From now on, we will assume that $\beta$ is no less than $\frac{2M}{k(2d-k+1)}$.

\section{Storage Cost Minimization} \label{sec:LP}
We solve the  linear programming problem in Theorem~\ref{thm:LP} by considering four different cases: (A) $n_1\geq k$ and $n_2\geq k$, (B) $n_1\geq k$ and $n_2 < k$, (C) $n_1 < k$ and $n_2 \geq k$, and (D) $n_1 < k$ and $n_2 < k$.

\subsection{Case A: $n_1\geq k$ and $n_2\geq k$}

When both $n_1$ and $n_2$ are larger than or equal to $k$, the two inequalities in~\eqref{eq:LP1} and~\eqref{eq:LP2} can be written as
\begin{align}
M &\leq m \alpha_1 + \theta_m , \text{ and} \label{eq:A1}\\
M &\leq m \alpha_2 + \theta_m . \label{eq:A2}
\end{align}
The region defined by these two inequalities is the intersection of two half-planes, which can be obtained by translating the first quadrant in the $\alpha_1$-$\alpha_2$ plane diagonally along the 45-degree line $\alpha_1=\alpha_2$.


\begin{theorem}
For $\beta \geq \frac{2M}{k(2d-k+1)}$, we have
\[
\alpha_1^*(\beta) = \alpha_2^*(\beta) =
\max_{1\leq m \leq k} (M - \theta_m)/m.
\]
\end{theorem}

\begin{proof}
Taking all constraints~\eqref{eq:A1} and~\eqref{eq:A2}, for $m=1,2,\ldots, k$ into consideration, the feasible region is in the form $\{(\alpha_1,\alpha_2):\, \alpha_1 \geq \mu \text{ and } \alpha_2 \geq \mu\}$, where $\mu$ is the maximum value as defined in the theorem. No matter what the costs $C_1$ and $C_2$ are, (provided that they are positive) the optimal solution to the linear programming is at the corner point of the feasible region, namely $(\alpha_1^*,\alpha_2^*) = (\mu,\mu)$.
\end{proof}

In the case where $n_1$ and $n_2$ are both larger than or equal to $k$, we see that the optimal storage allocation is to put the same amount of data in both type~1 and type~2 nodes. The storage costs of the two types of nodes do not matter.

\subsection{Case B: $n_1\geq k$ and $n_2<k$}
For $m=1,2,\ldots, k$, the two inequalities in~\eqref{eq:LP1} and~\eqref{eq:LP2} can be written as
\begin{align*}
m\alpha_1 &\geq M - \theta_m,\\
(m-q_m) \alpha_1 + q_m \alpha_2 &\geq M - \theta_m,
\end{align*}
where $ q_m \triangleq \min\{m,n_2\}$. These two inequalities define an infinite polyhedral region. For $m=1,2,\ldots, k$, let $\mathcal{R}_m$ be the region
\begin{align*}
 \mathcal{R}_m \triangleq &\{ (\alpha_1,\alpha_2)\in\mathbb{R}_+^2:\,
 m  \alpha_1  \geq M - \theta_m, \\
 & \quad (m-q_m) \alpha_1 + q_m\alpha_2 \geq M - \theta_m\},
\end{align*}
The feasible region of the linear program is thus the intersection of $\mathcal{R}_1$, $\mathcal{R}_2, \ldots, \mathcal{R}_k$.
The corner point of the region $\mathcal{R}_m$ can be obtained by solving the two equations obtained by setting the inequalities to equalities, and has coordinates
\[
\alpha_1=\alpha_2 = (M - \theta_m)/m.
\]
In other words, for $m=1,2,\ldots, k$, the corner point of $\mathcal{R}_m$ lies on the line $\alpha_1 = \alpha_2$ in the $\alpha_1$-$\alpha_2$ plane.

\begin{figure}
  \centering
  \scalebox{0.41}{\includegraphics{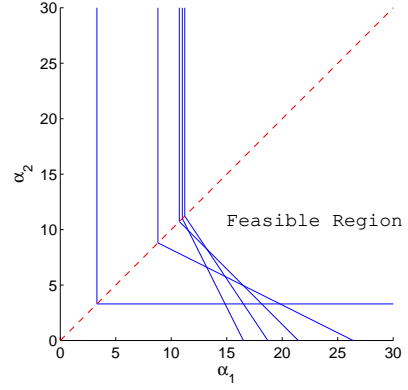}}
  \caption{An example of the feasible region in the linear program}
  \label{fi:region}
\end{figure}

An example of the feasible region is shown in Fig.~\ref{fi:region}. The horizontal and the vertical axes are $\alpha_1$ and $\alpha_2$ respectively. The parameters of the distributed storage system are $n_1=8$, $n_2=2$, $d=8$, $k=6$, $M=66$ and $\beta=3.3$. The region to the right and above all lines is the feasible region. The dashed line indicates the 45-degree line $\alpha_1=\alpha_2$. The optimal point is one of the vertices of the feasible region. The choice of the vertex which minimizes the storage cost depends on the ratio $C_1n_1/(C_2n_2)$, i.e., the slope of the objective function.

We can observe from Fig.~\ref{fi:region} that if the cost $C_1$ is much greater than $C_2$, then the optimal point always lies on the line $\alpha_1=\alpha_2$, i.e., $\alpha^*_1(\beta) = \alpha^*_2(\beta)$ for all~$\beta$.

Case C is similar to Case B. The feasible region of case C can be regarded as the mirror image of the feasible region of case B with respect to the line $\alpha_1=\alpha_2$. We therefore skip the discussion on Case~C.

\subsection{Case D: $n_1 < k$ and $n_2 < k$}
The feasible region  of the linear program in Theorem~\ref{thm:LP} is bounded by
\begin{align*}
p_m\alpha_1 +(m-p_m)\alpha_2&\geq M - \theta_m,\\
(m-q_m) \alpha_1 + q_m \alpha_2 &\geq M - \theta_m,
\end{align*}
for $m=1,2,\ldots, k$, where $q_m$ is defined as in the previous section and $p_m \triangleq \min\{m,n_1\}$.
The feasible region is the intersection of
\begin{align*}
 \mathcal{R}_m \triangleq &\{ (\alpha_1,\alpha_2)\in\mathbb{R}_+^2:\,
 p_m  \alpha_1 + (m-p_m)\alpha_2 \geq M - \theta_m, \\
 & \quad (m-q_m) \alpha_1 + q_m\alpha_2 \geq M - \theta_m\}
\end{align*}
for $m=1,2,\ldots, k$.
As in Case B, we can show that for $m=1,2,\ldots, k$, the vertex of the polyhedral region $\mathcal{R}_m$  lies on the line $\alpha_1=\alpha_2$ in the $\alpha_1$-$\alpha_2$ plane.

\section{Tradeoff between Storage Cost and Repair-Bandwidth} \label{sec:curve}
Explicit formulae for $\alpha_1^*(\beta)$, $\alpha_2^*(\beta)$ and $C_S^*(\beta)$ can be found, but due to space limitations, we do not type the formulae in this paper.

To illustrate the tradeoff between storage cost and repair-bandwidth, we consider a distributed storage system with parameters used in Fig.~\ref{fi:region}: $n_1=8$, $n_2=2$, $d=8$, $k=6$, $M=66$. The minimum repair-bandwidth is $2Md/(k(2d-k+1)) = 16$. We fix the cost $C_1$ for the storage nodes of type 1 to be 1, and increase $C_2$ from 0.2 to 1.8, with step size 0.4. For each value of $C_2$ we plot $C_S^*(\beta)$ for $d\beta$ from 16 to~32. The resulting curves are shown in Fig.~\ref{fi:tradeoff}. The curve in the middle corresponds to $C_1=C_2=1$. This reduces to the case in~\cite{DGWR2010} where the costs of both types of nodes are the same.

\begin{figure}
  \centering
  \scalebox{0.4}{\includegraphics{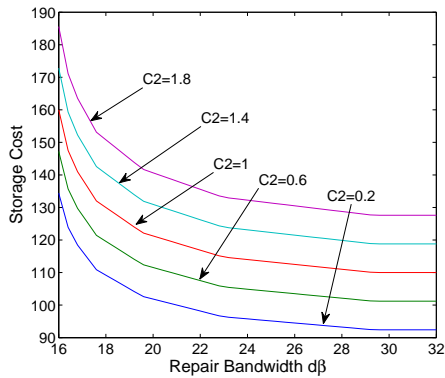}}
  \caption{Storage Cost and repair-bandwidth Tradeoff, $C_1=1$.}
  \label{fi:tradeoff}
\end{figure}

\section{Conclusion} \label{sec:Conclusion}
In this paper, we aim at seeking an optimal storage allocation that minimizes the storage cost in distributed storage systems. Specifically, we focus on the network with two types of storage nodes, each having a different storage cost. We demonstrate that the minimization problem can be solved as a linear programming problem. It is shown that the feasible region can be determined by analyzing the min-cut constraints of the corresponding information flow graph. The minimum storage cost can be achieved at the corner points. Moreover, the tradeoff between the storage cost and repair-bandwidth is established. Our method can be extended to more general cases, in which the storage costs of all storage nodes are not the same.

We can implement coding scheme and repair protocol for distributed storage system with storage cost by using random linear network coding over a finite field.  The packets transmitted from a surviving storage node to the newcomer are a linear combination of the data in the memory of the surviving storage node. If we apply existing code construction methods from linear network coding to distributed storage system, the required finite field size may be unbounded. It is because the finite field size requirement is a monotonically increasing function of the number of data collectors, which may be unbounded. To make sure that the regeneration process will be successful after arbitrarily many stages of repairs, it is important to show that the finite field size requirement is upper bounded by some constant. How to construct linear network code for distributed storage system with storage cost is an interesting direction for future studies.

\bibliographystyle{IEEEtran}

\bibliography{DStorage}

\end{document}